\documentclass[reprint,amsmath,amssymb,aps,pra,onecolumn]{revtex4-2}
\usepackage{array}[=2016-10-06]
\usepackage{dcolumn}
\usepackage{bm}
\usepackage{tabularx}
\usepackage{relsize}
\usepackage[utf8]{inputenc}
\usepackage[english]{babel}
\usepackage[T1]{fontenc}
\usepackage{hyperref}
\usepackage{graphicx}
\usepackage{etoolbox}
\usepackage{amsthm}
\usepackage{amsmath} 
\usepackage{amssymb}
\usepackage[colorinlistoftodos]{todonotes}
\usepackage{ulem}

\newtheorem{theorem}{Theorem}
\newtheorem{corollary}[theorem]{Corollary}
\newtheorem{lemma}[theorem]{Lemma}
\newtheorem{definition}[theorem]{Definition}
\newtheorem{remark}[theorem]{Remark}
\newtheorem{example}[theorem]{Example}

\DeclareMathOperator{\erfc}{erfc}
\makeatletter

\newcommand{\LDIrecordcite}[1]{%
  \ifx\@footnotetext\TX@trial@ftn\else
    \ifcsundef{LDIcount@#1}{%
      \csgdef{LDIcount@#1}{1}%
      \hypertarget{LDIcite.#1.1}{}%
    }{}%
  \fi
}

\AtBeginDocument{%
  \NewCommandCopy\LDIoriginalcite\cite
  \RenewDocumentCommand{\cite}{s o o m}{%
    \forcsvlist{\LDIrecordcite}{#4}%
    \IfBooleanTF{#1}{\def\LDIcallcite{\LDIoriginalcite*}}%
      {\let\LDIcallcite\LDIoriginalcite}%
    \IfNoValueTF{#2}{\LDIcallcite{#4}}%
      {\IfNoValueTF{#3}{\LDIcallcite[#2]{#4}}%
        {\LDIcallcite[#2][#3]{#4}}}%
  }%
}

\makeatother

\begin{document}

\preprint{APS/123-QED}

\title{Local-dimension-invariant Low-density parity-check stabilizer codes}

\author{Md. Shahinul Islam*}
\email{misla82@uic.edu}
\affiliation{Department of Electrical and Computer Engineering,
                    University of Illinois Chicago,
                   Chicago, Illinois, 60607}

\author{Aaryan Reza*}
\email{aaryanreza@ucla.edu}
\affiliation{Department of Electrical and Computer Engineering,
                    University of Illinois Chicago,
                   Chicago, Illinois, 60607}

\author{Lane G. Gunderman}
 \email{lanegunderman@gmail.com}
\affiliation{Department of Electrical and Computer Engineering,
                    University of Illinois Chicago,
                   Chicago, Illinois, 60607}

\date{\today}

\begin{abstract}
The protection of quantum information in systems leveraging more than two basis states (qudits) is an ongoing challenge. A tool for transforming the plethora of qubit stabilizer codes into qudit stabilizer codes is the local-dimension-invariant (LDI) extension to the stabilizer formalism. A drawback of the LDI method is that generally in order to ensure the distance of the resultant code is at least preserved the newly selected local-dimension will need to be any prime bigger than a large cutoff value. In this work we consider the repercussions of this large cutoff value, then by focusing on quantum low-density parity-check (qLDPC) codes we are able to bring down this cutoff value. Notably, we are able to reduce it in some cases to the initial local-dimension and for arbitrary CSS qubit-derived codes to singly exponential in $d$ with no dependence on the number of logicals $k$. This reduces the size of the caveat in LDI forms in their distance preservation condition.
\end{abstract}

\maketitle

\section{Introduction}\label{intro}

The protection of quantum information is currently a keystone to unlocking the potential of quantum computers. The most promising approach for attempting to correct errors induced from noise thus far has been stabilizer codes \cite{gottesman1996class,gottesman1997stabilizer}. These are the quantum analog of classical additive codes. The stabilizer formalism was initially designed for qubit quantum systems, but since has developed into a tool for systems with various local-dimensions. Having a method for transferring known stabilizer codes into forms that work for other local-dimensions has various uses, including in the construction of rotor codes \cite{vuillot2024homological}, coherent state codes \cite{xu2024letting}, GKP-like codes \cite{conrad2022gottesman,conrad2023good}, among other locations \cite{gunderman2024stabilizer}. One tool for transforming known stabilizer codes into forms which commute regardless of the local-dimension is the \textit{local-dimension-invariant} (LDI) framework. This tool takes an already known stabilizer code and makes it so that the generators of the code commute regardless of the local structure of the registers. This work improves upon the distance caveats of this work, particularly focusing on sparse quantum codes. By focusing on sparse codes the super exponential local-dimension requirement to ensure the distance is able to be drastically reduced. We find some constructions whereby the cutoff value is either the initial local-dimension, or only scales singly exponentially in the distance while being independent of the number of physical and logical registers. Besides this, we provide some new interpretations for the cutoff value as a precision requirement measure.

This work is organized as follows. We begin by laying out traditional definitions before we move on to definitions related to local-dimension-invariant codes. In section \ref{meanings} we discuss an interpretation of the $p^*$ (distance preservation bound) in terms of squeezing and precision required in experiments, followed by a brief discussion for why traditional quantum weight reduction cannot be directly applied. In section \ref{qubitcase}, we discuss a special construction for LDI forms derived from qubit codes with certain structure which keeps the maximal entry to absolute value of $1$, which can permit for a far lower cutoff bound. Following that, in section \ref{qldpcres} we discuss various tightenings of the cutoff bound utilizing low-density parity-check (LDPC) properties of the initial code. Notably these results do not demand $O(1)$ sparsity, and can greatly reduce the bound if only the check weight or register degree are sparse. In section \ref{topo} we discuss some new topologic LDI codes, both boundary expanding and thus having their distance and logical count preserved from topologic structure as well as hyperbolic codes that do not have boundary expansion but mild cutoff bounds. We then close out the paper with future directions that seem like the best candidates for further research.

\section{Definitions}\label{defns}


To begin, we lay out the traditional qudit stabilizer code definitions. This will ground our further results. For a good overview of qudit codes and many constructions of them, this paper is suggested \cite{ketkar2006nonbinary}. For extensions beyond prime and prime-power qudits, we suggest \cite{gunderman2024stabilizer}. As this work will be working with qubits, qudits, continuous-variable codes, among other options, we will use the general term \textit{register} where we wish to not specify the local-dimension \cite{watrous2018theory}. As such, qubit is used when the result solely applies to qubits, while register is used otherwise. Each register can be considered a separate location whereby quantum mechanical operations are being performed, and a collection of registers forms a full quantum computer.

As our central operations, we define the qudit version of the traditional Pauli group. This is able to be generated, up to scalar coefficients, by multiplication of two elements.

\begin{definition}
    The qudit Pauli group over some finite local dimension $q$ can be generated by: \begin{equation}
        X_q |j\rangle = |j+1 ~ mod ~ q\rangle, ~~ Z_q|j\rangle = \omega_q^j, ~~ \omega_q = e^{2\pi i /q}
    \end{equation}
\end{definition}

When carrying over to infinite settings, we can utilize the Heisenberg group either in the continuous case or a discrete version. We are studying stabilizer codes, which is able to be specified by a set of commuting Pauli operators that generate a static Hilbert subspace, which specifies our encoded information.

\begin{definition}
    A stabilizer code $S$ with stabilizers $s_i$ is specified by its parameters $[[n,k,d]]_q$:

    \begin{itemize}
        \item{} $n$: the number of physical registers used
        \item{} $k$: the number of logical registers represented
        \item{} $d$: the distance of the code, given by the lowest weight undetectable Pauli error. These are the Pauli operators that commute with all generators of the code but is not in the group generated by these operators themselves.
        \item{} $q$: the number of values that each register can take, called the local-dimension.
    \end{itemize}
\end{definition}

It will aid our discussions to regularly represent our Pauli operators in terms of their symplectic representation \cite{lidar2013quantum}, here referred to as the $\phi$ representation so that a subscript can be added to specify the local-dimension \cite{gunderman2020local}.

\begin{definition}
    The $\phi$ representation of a qudit operation is given by the linear surjective map \begin{equation}
        \phi_q = \mathbb{P}_q^n \to \mathbb{Z}_q^{2n}
    \end{equation} This carries a $n$ register Pauli operator from $\mathbb{P}_q^n$ to a $2n$ entried symplectic vector. This mapping can equivalently be stated as \begin{equation}
        \phi_q (I^{i-1}\otimes (X^aZ^b)\otimes I^{n-i}) = [0^{i-1}~a~0^{n-i}|0^{i-1}b~0^{n-i}]
    \end{equation}
    This puts the power of the $i$-th $X$ operator in the $i$-th position and the $j$-th $Z$ operator in the $n+j$-th position. Furthermore, it is a homomorphism with $\phi_q(s_is_j) = \phi_q(s_i)\oplus\phi_q(s_j)$ where $\oplus$ is component-wise addition mod q. Setting $q=\infty$ gives the case over $\mathbb{Z}$.
\end{definition}

While the prior definition provided a way to represent Pauli operators, we must also specify a way to determine whether a pair of Pauli operators commute. For this, we need the symplectic inner product. This computes the power of the phase induced from switching the order of two Pauli operators. For a pair of Pauli operators in the $\phi_q$ representation, we have $(\vec{a}_1|\vec{b}_1)\odot(\vec{a}_2|\vec{b}_2)=\vec{a}_1\cdot \vec{b}_2-\vec{a}_2\cdot \vec{b}_1\mod q$. When this is congruent to zero the two operators commute as a net phase of an integer power of $e^{2\pi i}$ is induced from swapping the order of the operators.

Some of our results pertain only to CSS stabilizer codes. This is a slight restriction upon all possible stabilizer codes, but is a common restriction that is utilized.

\begin{definition}
    A Calderbank-Shor-Steane (CSS) code is a code where checks for bit and phase flips are decoupled, such that stabilizers checking for bit (phase) flips are purely $X$ ($Z$) operators. This leads to the resulting stabilizer matrix structure \begin{equation} 
    \phi(S_{CSS}) =\begin{pmatrix}
    H_{X} & 0 \\
    0 & H_{Z}
    \end{pmatrix}
    \end{equation}
\end{definition}

We also state that the stabilizer check weight is given by the maximum of the minimal Pauli weight stabilizer generators, while the register degree is the maximal number of times that a given register is queried across a set of generators. We may restrict these terms to just considering $X$ checks or just $Z$ checks. With this, we have our most essential definitions for stabilizer codes.


We are examining stabilizer codes which can be applied to systems of various local-dimension choices. Toward this end, we recall the tool of local-dimension-invariant (LDI) forms. A stabilizer code is said to be LDI if the generators commute regardless of the local-dimension. An early example of this is the $5$-register code \cite{chau1997five}. Another example is the code generated by $\langle XX,ZZ^{-1}\rangle$. More generally, we have:

\begin{definition}[Local Dimension Invariance] A quantum error correcting code is local-dimension-invariant (LDI) if its stabilizer generators commute over any choice of local dimension $q$. \begin{equation}
    \phi_q(s_i) \odot \phi_q(s_j) \equiv 0 ~mod~q~\forall i,j\in [n-k],q\in \mathbb{Z}_{>0}
\end{equation}
    where $\odot$ is the symplectic product. This is an equivalent condition to commutation over $\mathbb{Z}$ and thus can be reduced to \begin{equation}
        \phi_\infty (s_i) \odot \phi_\infty(s_i) = 0~ \forall i,j \in [n-k]
    \end{equation}
\end{definition}

This formalism was first introduced in \cite{gunderman2020local} and has since been expanded upon. In the initial work a prescriptive method was provided for turning any stabilizer code over a finite field into an LDI form. This is not a unique procedure, but helps with showing the validity for arbitrary stabilizer codes. From the very definition of LDI forms, we see that they form a valid set of commuting Pauli operators regardless of the local-dimension. This suffices for obtaining the parameters $[[n,k]]$, however, to relate the distance of the code for some new local-dimension to that of the code of the original local-dimension, further results are required. In order to ensure the distance of the code for some new local-dimension $q'$ is at least as large as the initial local-dimension code's distance, $q'$ must generally be larger than some cutoff value $p^*$, which can be quite a large value. Reducing this cutoff value for promising the distance of the resulting code forms the bulk of the motivation of this work.


\section{Interpretations of $p^*$ bound}\label{meanings}

\subsection{Precision and squeezing related bounds}

Bosonic registers have various representations, including number basis, phase space, and the coherent basis. These representations have then led to various proposals for encoding quantum information, including Binomial codes \cite{michael2016new}, Gottesman-Kitaev-Preskill (GKP) encoding \cite{gottesman2001encoding} and more general GKP-like encodings \cite{gunderman2024stabilizer,conrad2022gottesman,conrad2023good}, CV analog codes \cite{braunstein1998error,lloyd1998analog}, cat qubits and their further extension into tiger codes \cite{xu2024letting}, among various options. For these varieties of additive codes, there is generally a requirement of the symplectic products of the checks being $0$ or a multiple of an integer, for which the LDI form is particularly well suited for satisfying. Due to finite energy conditions, the analysis may reduce the infinite precision domains into Gaussian envelopes, or at times, finite order of generators. A clear example is the analysis carried out in the original GKP paper \cite{gottesman2001encoding}, whereby the delta peaks are transformed into Gaussian peaks, and then an overall Gaussian envelope is applied. In the ideal GKP code, shifts by $\sqrt{2\pi}$ an even number of times should always leave the codespace perfectly preserved, but due to finite energy this is not the case, instead the codespace must be re-centered in phase space regularly.

In the original GKP work, the possibility of encoding qudits into phase space was also discussed. Rather recently this has been implemented experimentally up to a break-even point for qutrits and ququarts \cite{brock2025quantum}. As one increases the local-dimension $q$ there are a number of competing factors that must be considered. One is that if one makes the code translationally invariant up to $\sqrt{2\pi/q}$ shifts, the size of phase space shift that can be recovered from shrinks. As a knock-on effect this means that the level of squeezing that would be required to make this practical would need to increase as the local-dimension does. This is something that we quantify. Second, if one increases the spacing of the underlying lattice, this corresponds to increasing the average excitation count, which then increases the rate of excitation dissipation, in turn speeding up the rate of Gaussian blurring in the phase space \cite{brock2025quantum}. That direction would also be wise to quantify more explicitly so that experimentally motivated limits on the possible local-dimension can be set, but we focus here on the former constraint.

We now proceed to an analysis of the expected level of squeezing required to achieve a given local-dimension assuming a spacing of $\sqrt{2\pi/q}$ between the peaks for the different bases. We will make the choice that $[q,p]=i$ for the quadratures $q,p$ and so the vacuum standard deviation is $1/\sqrt{2}$. Replacing the delta peaks by Gaussians with variance $\sigma_q^2$ results in a decibel expression for the ratio of the variances as:
    \begin{equation}
        s_{dB}=-10\log_{10}(\sigma_q^2/\sigma_{vacuum}^2)=-10\log_{10}(2\sigma_q^2).
    \end{equation}
As a warning to readers from different communities this most closely follows the form used in GKP error analyses, but varies from some other communities. The general analysis will still hold from here, although some multiplicative factors may deviate. We summarize the various relations between the variance of the peaks and the squeezing parameters in the table below:

\begin{center}
\scriptsize
\renewcommand{\arraystretch}{1.05}
\setlength{\tabcolsep}{4pt}

\begin{tabularx}{\linewidth}{|X|X|X|}
\hline

\textbf{Source convention}
&
\textbf{Variance $\sigma_{q}^{2}$}
&
\textbf{Final squeezing formula}
\\
\hline

Gaussian finite-energy envelope operator
$E_{\Delta}=e^{-\Delta^{2}\hat{n}}$
&
Exact:
$\sigma_{q}^{2}=\tfrac{1}{2}\tanh(\Delta^{2})$;
small-$\Delta$:
$\sigma_{q}^{2}\approx\tfrac{\Delta^{2}}{2}$\,$^{\ast}$
&
Exact:
$s=-10\log_{10}\tanh(\Delta^{2})$;
for $\Delta^{2}\ll 1$:
$s\approx-20\log_{10}\Delta$
\\
\hline

$\sigma_{\mathrm{gkp}}$ quoted directly, as in the Gaussian
random-displacement model used by Noh--Chamberland \cite{Noh2020FaultTolerant}
&
No conversion required:
$\sigma_{q}^{2}=\sigma_{\mathrm{gkp}}^{2}$
&
$s=-10\log_{10}(2\sigma_{\mathrm{gkp}}^{2})$
\\
\hline

Terhal--Weigand effective squeezing parameter $\Delta_{q}$,
defined through stabilizer sharpness; their vacuum has
$\Delta_{q}=1$ \cite{Duivenvoorden2017DisplacementSensor}
&
For a Gaussian peak,
$\sigma_{q}=\Delta_{q}\sigma_{\mathrm{vac}}$;
since $\sigma_{\mathrm{vac}}=1/\sqrt{2}$,
$\sigma_{q}^{2}=\tfrac{\Delta_{q}^{2}}{2}$
&
$s=-20\log_{10}\Delta_{q}$
\\
\hline

Literal Gaussian wavefunction width $w_{q}$,
$\psi(q)\propto e^{-(q-q_{0})^{2}/(2w_{q}^{2})}$
&
$\sigma_{q}^{2}=\tfrac{w_{q}^{2}}{2}$
&
$s=-10\log_{10}(w_{q}^{2})
=-20\log_{10}w_{q}$
\\
\hline

Literal Gaussian probability width $w_{q}$, when a paper writes
$P(q)\propto e^{-(q-q_{0})^{2}/(2w_{q}^{2})}$
&
$\sigma_{q}^{2}=w_{q}^{2}$
&
$s=-10\log_{10}(2w_{q}^{2})$
\\
\hline

Ordinary squeezing parameter $r$, with the $q$ quadrature squeezed
&
$\sigma_{q}^{2}=\tfrac{1}{2}e^{-2r}$
&
$s=-10\log_{10}(e^{-2r})
=\tfrac{20}{\ln 10}\,r
\approx 8.686\,r$
\\
\hline

\end{tabularx}
\end{center}

{\footnotesize $^{\ast}$ For a small number $x$, $\tanh x=x-\tfrac{x^{3}}{3}+\tfrac{2x^{5}}{15}-\cdots$; here $x=\Delta^{2}$, so $\tanh(\Delta^{2})=\Delta^{2}-\tfrac{\Delta^{6}}{3}+\cdots\approx\Delta^{2}$ when $\Delta\ll 1$.}

A recent superconducting-cavity experiment prepared GKP qutrit and ququart target states with a finite-energy parameter reported as $\Delta=0.32$ in its chosen convention and demonstrated error-corrected memories beyond break-even~\cite{brock2025quantum}. Under the local-Gaussian conversion for that state model, $\sigma_{q}^{2}\approx\tfrac{\Delta^{2}}{2}=\tfrac{0.32^{2}}{2}=0.0512$, so $\sigma_{q}\approx 0.2263$, and the corresponding variance-based squeezing is $s_{\mathrm{dB}}=-10\log_{10}(0.0512/0.5)\approx 9.90$~dB. Inverting the squeezing formula for the width gives
\begin{equation}
\sigma_{q}^{2}=\tfrac{1}{2}\,10^{-s_{\mathrm{dB}}/10},\qquad \sigma_{q}=\tfrac{1}{\sqrt{2}}\,10^{-s_{\mathrm{dB}}/20}.
\end{equation}

For a square-lattice $p$-dimensional GKP qudit, choose the logical displacement spacing $\Delta_{p}=\sqrt{2\pi/p}$; the stabilizer period is $L_{p}=p\Delta_{p}=\sqrt{2\pi p}$, and a convenient ideal $Z$-basis representation is
\[
\lvert\bar{j}\rangle\propto\sum_{n\in\mathbb{Z}}\,\lvert q=(j+np)\Delta_{p}\rangle,\qquad j=0,1,\dots,p-1.
\]
The nearest neighboring logical value is one spacing away, so the halfway boundary is $a_{p}=\tfrac{\Delta_{p}}{2}=\sqrt{\tfrac{\pi}{2p}}$, the error-correction radius for nearest-neighbor decoding. Both scale as $\Delta_{p},a_{p}\propto p^{-1/2}$: a larger dimension packs more logical values into phase space and moves their decision boundaries closer together.

Modeling the net shift in one quadrature as $\xi\sim\mathcal{N}(0,v)$ with $v=\sigma_{\mathrm{eff}}^{2}$ and probability density function $f(\xi)=\tfrac{1}{\sqrt{2\pi v}}\exp\big(-\tfrac{\xi^{2}}{2v}\big)$. A nearest-boundary logical mistake occurs when the shift crosses either boundary $\pm a_{p}$; for a zero-mean Gaussian, $\Pr(\lvert\xi\rvert>a)=\erfc\big(a/\sqrt{2v}\big)$. Substituting $a_{p}=\sqrt{\pi/2p}$,
\[
\epsilon(v)=\Pr\big(\lvert\xi\rvert>a_{p}\big)\approx\erfc\sqrt{\frac{\pi}{4pv}}.
\]

We may invert this expression to obtain a bound on the required squeezing for a given permission error rate. The design problem runs in reverse: fix a target $\epsilon(v)\le\epsilon$ and ask how small $v$ must be. From $\epsilon=\erfc\sqrt{\pi/4pv}$, applying the inverse complementary error function gives $\erfc^{-1}(\epsilon)=\sqrt{\pi/4pv}$; with $\kappa_{\epsilon}=\sqrt{2}\,\erfc^{-1}(\epsilon)$ and squaring,
\[
\frac{\kappa_{\epsilon}^{2}}{2}=\frac{\pi}{4pv}\quad\Longrightarrow\quad v\le v_{\max}(p,\epsilon)=\frac{\pi}{2\kappa_{\epsilon}^{2}p}.
\]
For fixed $\epsilon$, $v_{\max}\propto 1/p$, equivalently $\sigma_{\max}\propto 1/\sqrt{p}$, matching the $a_{p}\propto p^{-1/2}$ shrinkage of the boundary.

Using this, we derive the local-dimension--squeezing law, relating the squeezing required to the number of levels that can be maximally utilized. Suppose finite peak width is the only contribution to the effective variance, $v=\sigma_{\mathrm{GKP}}^{2}=\tfrac{1}{2}10^{-s_{\mathrm{dB}}/10}$ by~(1). The required condition $\tfrac{1}{2}10^{-s_{\mathrm{dB}}/10}\le\tfrac{\pi}{2\kappa_{\epsilon}^{2}p}$ gives $10^{-s_{\mathrm{dB}}/10}\le\tfrac{\pi}{\kappa_{\epsilon}^{2}p}$; taking $\log_{10}$ and multiplying by $-10$ reverses the inequality:
\[
s_{\mathrm{req}}(p,\epsilon)=10\log_{10}\frac{\kappa_{\epsilon}^{2}p}{\pi}=10\log_{10}p+C_{\epsilon},\qquad C_{\epsilon}=10\log_{10}\frac{\kappa_{\epsilon}^{2}}{\pi},
\]
separating the dimension-dependent and error-dependent pieces. We see that the additional squeezing required scales logarithmically with the number of levels utilized. While it is experimentally challenging, this could serve as a valuable avenue for increasing the quantum information storage per quantum register.

Using the 9.5~dB GKP-qubit benchmark from the grid-state experiment~\cite{campagne2020quantum} as the reference point. In the nearest-boundary model this gives $\epsilon_{\mathrm{ref}}=1.83\times 10^{-4}$ for $p=2$. Asking what squeezing preserves the same geometric crossing probability at higher dimension, with the error target fixed,
\[
s_{\mathrm{req}}(p)=9.5\ \mathrm{dB}+10\log_{10}\frac{p}{2},
\]
which gives the following table of results:

\begin{center}\small
\begin{tabular}{ @{}l ccccccc@{}}
\hline \hline
$p$ & 2 & 3 & 4 & 5 & 7 & 11 & 59\\
\hline
$s_{\mathrm{req}}$ (dB) & 9.50 & 11.26 & 12.51 & 13.48 & 14.94 & 16.90 & 24.20\\
\hline \hline
\end{tabular}
\end{center}

As $p$ increases, more logical states must fit into the same phase-space cell. Therefore, neighboring logical regions become closer together. A Gaussian peak must consequently become narrower to maintain the same probability of crossing a decoding boundary.

We now turn our attention to an alternative encoding scheme and the implications of our cutoff bound there. We now consider systems where at least one of the quadrature axes is representable by some real value. This includes rotor encodings, tiger encodings, as well as CV encodings. The distance of these codes, when expressed in Pauli Hamming distance, assume infinite precision in measurements through the distance of the underlying codes being computed over the reals or integers. In reality, infinite precision is not attainable. Let's consider the case where a quadrature axis is a real number in $[0,1)$ and let our precision be such that the smallest separation that is able to be discerned is $\Delta$. Then the true local-dimension is more similar to $1/\Delta$ since operators will, up to precision, appear to have order $1/\Delta$. As such, the distance of the code must be computed with local-dimension at most $1/\Delta$. Then if $1/\Delta>p^*$, we know the distance of the resultant code will be at least that of the original code. This then implies a bound on the precision required to immediately apply the distance that is computed over the integers instead of computing it conditioned on the precision.


\subsection{Why doesn't standard quantum weight reduction work?}

Reducing the weights of our LDI forms serves a pair of purposes. First, as our bounds on the sufficient local-dimension to preserve the distance will be a function of the weights, having lower weights will permit us to lower these values, and thus reducing the set of local-dimensions which must have their distances computationally verified or verified through other formal methods. Secondly, there is the general advantages that occur from having lower weights. This generally improves second-order noise as well as typically improving decoding runtime \cite{tzitrin2021fault}. The topic of weight reduction has become of intense interest within qubit stabilizer research, resulting in very low qubit degrees and stabilizer weights--assuming an $O(1)$ start \cite{hastings2021quantum,he2025discovering,hsieh2025simplified,sabo2024weight}. The standard quantum weight reduction protocol utilizes repeated use of mapping cones.

The mapping cones can preserve commutation properties, as well as, with caveats, the parameters of the code while lowering the weights. Assuming an LDI form of a stabilizer code is used as input, an LDI version of the mapping cones follows immediately from a sign swap in the mapping chain \cite{sabo2024weight}. However, that is not sufficient to ensure that the parameters of the code are otherwise preserved, only that the symplectic products can be kept as zero. Among the challenges in making an LDI version of quantum weight reduction would be requiring high quality multi-graph, or weighted graph, expanders where the weight of the edges would correspond to the powers of the Heisenberg displacement operators. In the event that the initial LDI code only has entries from $\{0,\pm 1\}$, then quantum weight reduction should work with careful sign flips of the coning map to indicate an oriented map.

Instead of showing an LDI weight reduction procedure, this work attempts to reduce the distance preservation bound utilizing just the LDPC property of codes, in effect requesting $O(1)$ but not a particularly small number, and seeing how far that takes us.

\section{Result specialized to initially qubit stabilizer codes}\label{qubitcase}

As our inspiration for our new result, we look to a particular way of generating an LDI form that preserves CSS structure for the Steane code \cite{moorthy2023local}:
\begin{equation}
    \begin{bmatrix}
        1 &1&1&1&0&0&0&|&0&0&0&0&0&0&0\\
        0 &1&1&0&1&1&0&|&0&0&0&0&0&0&0\\
        0 &0&1&1&0&1&1&|&0&0&0&0&0&0&0\\
        0&0&0&0&0&0&0&|&1&-1&1&-1&0&0&0\\
        0&0&0&0&0&0&0&|&0&1&-1&0&-1&1&0\\
        0&0&0&0&0&0&0&|&0&0&1&-1&0&-1&1\\
    \end{bmatrix}.
\end{equation}
Notice that here, $\mathcal{X} = \mathcal{Z}$ up to polarity after converting into LDI form, and of four non-zero elements per row, two have their signs flipped. This same pattern can be seen in the LDI form for the qubit quantum Hamming family. Since the overlap between different qubit stabilizers is even in weight by design, this can intuitively be extended to binary seed codes with certain structure.

Generally for a local-dimension $q$ stabilizer code the best that we can hope for with an LDI form of a code is to have the maximal entry $B$ bounded by $B=q-1$, besides from random chance aligning well. For qubit codes this corresponds to stabilizer generators with entries only from $\{0,\pm 1\}$. We present next a sufficient condition to ensure that there is an LDI form of a qubit stabilizer code with $B=1$.

For a CSS stabilizer code, written in symplectic form, we define the conflict graph for a generator $s_{x,i}$ from $H_X$, denoted $G_C(s_{x,i})$ as follows. Create a vertex for each support in $s_{x,i}$ and call this graph $G_0$. For each generator $s_{z,j}$ from $H_Z$ create an edge between vertices $k,l$ if those registers are in the support of $s_{z,j}$. This will generate a complete graph on the registers in their intersection, which becomes our new graph $G_1$. Next, use $G_1$ as the beginning graph for the next generator $s_{z,j'}$ from $H_Z$. Repeat this process for all generators from $H_Z$. Naturally this construction can be performed where we swap the roles of $H_X$ and $H_Z$. The following result uses special conditions on the conflict graph to ensure an automatic assignment of signs in the LDI form.

\begin{theorem}[Sign swap algorithm]\label{sswap}
Let $S$ be a qubit CSS stabilizer code. Further, if there is a set of generators for the code such that the conflict graph for each generator, $G_C(s_{x,i})$, is bipartite and for each generator $s_{x,i}$ in $H_X$ $wt(s_{x,i})\geq |G_C(s_{x,i})|+1$, then there exists an LDI form for the code such that the entries are all $\{0,\pm 1\}$ and the supports of all generators are unchanged. If this condition is satisfied for $H_Z$ and not $H_X$, the result holds with the replacement.
\end{theorem}



The goal of the proof is to use the $2$-colorability of bipartite graphs along with sufficient degrees of freedom to ensure that sign flips suffice to make the symplectic products of the resulting code $0$. This case is possible due to $\pm 1$ differing by a multiple of the size of the initial field ($\mathbb{Z}_2$), thus leaving the original code invariant under the initial field.

\begin{proof}
    We can arbitrarily pick $H_X$ or $H_Z$ to be signed. For this proof, we pick $H_X$. We see that for a stabilizer row $s$ we can sign any qubit independently, but since global sign flips can be factored out we have $wt(s)-1$ degrees of freedom. Let $G_C(s)$ be the conflict graph of generator $s$. If this graph is bipartite then we can $2$-color this graph with $1,-1$ to sign $H_X$ and yield an LDI stabilizer generator. A proof of commutation follows from the fact that any pair of qubits that would commute to $1\cdot1 + 1\cdot1 = 0 ~mod~2$ now commute to $1\cdot-1 + 1\cdot1 = 0$, leading to a total stabilizer overlap of $\sum(1\cdot-1 + 1\cdot1) = 0$. Since this is presumed to be true for all generators the resulting code will be LDI.

    If $H_X$ is not $2$-colorable, we can pick $H_Z$ to sign instead. However, if neither conflict graph is $2$-colorable then this result does not apply. Each sign flip consumes one degree of freedom. Therefore for a given code, a stabilizer row can be put into LDI form with just sign flips if
    \begin{equation}
    wt(s_i) \ge |G_C(s)| +1
    \end{equation}
    where $|G_C(s)|$ is the number of stabilizer rows of $H_Z$ $s$ has a unique overlap with.
\end{proof}

This result is applied to hyperbolic surface codes in Lemma \ref{sshyper} to obtain hyperbolic LDI codes with powers of only $\pm 1$ and $0$, providing a minimal valued maximal value in the $\phi$ representation for these codes. It is worth parsing the conditions of this result. Due to the constraints, we need generators to have an overlap of exactly $2$ (since qubits), but different generators from the dual checks can overlap on different pairs of qubits in our original check. This is a rather strong restriction, but does permit some codes to be put into an LDI form. Notably, this does not even permit the Steane code. For this, we present the following sufficient condition, which ultimately rephrases the problem of assigning signs to designing direct graphs composed of cycles. To this end, we define a condition, then proceed to prove that it suffices for generating an LDI form with $B=1$.

\begin{definition}[Cycle sign assignment condition]
Let $S$ be a qubit CSS stabilizer code, written in symplectic form. We define the hollowed conflict graph for a generator $s_{x,i}$ from $H_X$, denoted $G_C^0(s_{x,i})$, as follows. Create a vertex for each support in $s_{x,i}$ and call this graph $G_0^0$. For each generator $s_{z,j}$ from $H_Z$ create an edge between vertices $k,l$ if those registers are in the support of $s_{z,j}$. This will generate a complete graph on the registers in their intersection, which becomes our new graph $G_1$. Now we transform $G_1$ into a set of hollowed forms, whereby for each hollowed for we break the complete graph into a cycle graph on the supporting vertices. Then $G_1^0$ is the full set of all possible hollowed conflict graphs with that single generator. Next, use $G_1^0$ as the beginning graphs for the next generator $s_{z,j'}$ from $H_Z$. Repeat this process for all generators from $H_Z$. This resulting set is our hollowed conflict graph set for a generator $s_{x,i}$, written $G_C^0(s_{x,i})$.

For each cycle graph in $G_h^0$ that is added to the total hollowed conflict graph set, we will set an orientation so that vertex $l$ is connected to vertex $l\pm 1$ only. If the sign at that vertex is assigned to be $+1$ then the edge is replaced with a directed edge out from $l$ to $l+1$, while if the sign at that vertex is assigned as $-1$, then the edge is replaced with a directed edge in from $l+1$ to $l$. Classify the vertices then as either having a pair of incoming edges, "ii", a pair of outgoing edges, $oo$, the first incoming and the second outgoing "io", and the first outgoing and the second incoming as "oi". Consider all such possible assignments, where only some are LDI forms. Perform this process for each value of $h$ in the build up of the full set of hollowed conflict graphs. If there is an assignment such that for each cycle $|io|=|oi|$ then the resulting graph is called "cycle swap supporting".
\end{definition}

\begin{theorem}\label{swapassignment}
Let $S$ be a qubit CSS stabilizer code. If each generator $s_{x,i}$ has a cycle swap supporting conflict graph, then the resulting code has $B=1$, so that the resulting code only has entries of $\{0,\pm 1\}$.
\end{theorem}

\begin{proof}
Consider the condition for "cycle swap supporting" for a given generator $s_{x,i}$. Through the assignment of signs for generators $s_{z,j}$ there will be outgoing-outgoing vertices, which correspond to $+1-1$, incoming-incoming vertices, which correspond to $-1+1$, and then incoming-outgoing vertices and outgoing-incoming vertices. These last two correspond to $+1+1$ and $-1-1$, thus we need an equal count of these for the total symplectic product to be zero. This condition holding for each generator suffices for the code to be an LDI form, and the signs restrict the entries to $\{0,\pm 1\}$.
\end{proof}

This result does cover cases such as the Steane code and the full quantum Hamming family \cite{moorthy2023local}, however, this is not an automatic assignment of signs for general qubit CSS codes. This essentially provides a sufficient directed graph theoretic condition. In the ideal case this would also have some degrees of freedom related bounding, which would greatly improve the utility of this result. We now turn our attention to improvements to qLDPC stabilizer codes, permitting initial local-dimensions beyond binary.


\section{LDI qLDPC Results}\label{qldpcres}

To begin, we will show that a particular construction will always preserve its distance, and has a rather simple LDI form. A particularly nice code construction method, the hypergraph product construction \cite{tillich2013quantum}, rather immediately can be put into an LDI form, however, in order to ensure the distance is preserved we need the following short Lemma.

\begin{lemma}
Let $H_q$ be a classical error correcting code's parity check matrix with entries from $\mathbb{Z}_q$ and with parameters $[n,k,d]_q$. Then there is a transform that generates a similar code, $H_\infty$, so that $H_\infty\mod q=H_q$ and $H_\infty$ has parameters $[n,k,d]_Q$ for any integer $Q\geq q$.
\end{lemma}

\begin{proof}
The canonical form for a classical code's parity check matrix $H_q$ is given by $[I_{n-k}\ A]$. From this the codeword generator matrix $G$ can be given by $[A^T\ -I_k]$ so that $H_qG^T=0$. Promoting the entries in $H_q$ from $\mathbb{Z}_q$ to entries in $\mathbb{Z}_Q$ to form $H_{Q}$ will preserve the entries in the canonical form for both the parity check and the generator matrix, thus the distance is unaltered.
\end{proof}

This result is rather direct, however, is required to apply the hypergraph product construction and ensure the distance of the resulting code is at least as large as that of the HGP of the classical codes before dimensionally changing. Let $H_1$ be a classical error correcting code specified by $r_1$ many parity checks and with parameters $[n_1,k_1,d_1]$. Let $H_2$ be a classical error correcting code with parameters subscripted by $2$. Then a set of LDI generators can be obtained by constructing a stabilizer code in the following way: $H_X=(H_1\otimes I_{n_2}\ \ I_{r_1}\otimes H_2^T)$ and $H_Z=(I_{n_1}\otimes H_2\ \ -H_1^T\otimes I_{r_2})$. If we use our similar classical code, this will ensure that the distance is preserved always, thus generating LDI HGP codes with the distance preserved for all primes at least as large as that for the original code. We could even apply classical weight reduction on the seed codes and generate extremely sparse LDI codes with up to constant rate and distance scaling as $O(\sqrt{n})$ \cite{sabo2024weight}.

We now proceed to show a weight, degree, and maximal entry bound that holds for general stabilizer codes upon being put into an LDI form. This will aid in our tightening of a bound on the $p^*$ value ensuring the distance is preserved for LDI qLDPC codes.

\begin{lemma}
Let $S$ be a qLDPC stabilizer code with parameters $[[n,k,d]]_q$ with check weights at most $w$ and register degrees of at most $c$. Let $k\geq \max(w,c)$. Then this may be transformed into an LDI form with checks weights at most $w^2(1+(2c-1)/2)$ and register degree at most $(2c-1)(1+w^2/2)$, while the maximal entry, in absolute value, $B$ satisfies $B\leq (q-1)[1+(q-1)w^2]$.
\end{lemma}

Note that not every set of generators for a qLDPC code will be low-weight. So we must account for the increase in the weights for putting a stabilizer code into an LDI form. For generality we will not assume that the code is CSS and provide a bound for this case. The prior general result was $B\leq [2+(q-1)k](q-1)$, and so this provides an $O(1)$ bound for qLDPC codes even when they have a large number of logical operators. When this is paired with quantum weight reduction techniques such as \cite{hastings2021quantum,sabo2024weight,hsieh2025simplified}, this means that maximal entry $B$ will stay rather manageable, at roughly $37$ for any qubit derived LDI qLDPC code. In a later result we show that this can be improved further due to the CSS structure of such codes.

\begin{proof}
The stabilizer generators begin with check weights at most $w$ and register degree at most $c$. We will put the code into canonical form. To clear the first column, assuming swaps and single register Cliffords have been performed to put the first pivot there, we must add that row to $c-1$ rows, meaning that those $c-1$ rows have had their weights increased to up to $2w$. We will repeat this procedure for all $k$ pivots needed, but considering just the first row, at most $w-1$ other rows must be added to it, thus the weight is increased by $(w-1)w$ to a total of $w^2$. Then in total the resulting canonical form will have check weights of $w^2$, and register degrees of at most $c+(c-1)=2c-1$ (as the symplectic dual might always have the degree increased).

As the next step we will assume we have a canonical form for the a stabilizer code $[I_k\ X_2\ |\ \ Z_1\ Z_2]$ with weights at most $w'$ and register degrees at most $c'$. In the prescriptive process of generating LDI forms for the generators by adding a lower triangular matrix $L$ of the symplectic inner products over the integers, we must account for how the weights and degrees may increase. The $L$ matrix is $k\times k$, but due to the sparsity $L$ will also be sparse. To figure out how many nonzero entries a single row of $L$ could have, we must count the number of possible generators which that generator's row can have nontrivial support overlap. Since the row has weight $w'$ each entry can intersect with at most $c'w'$ many total entries. Since a single overlap is insufficient for the original code to commute, we must consume at least two entries per intersection, so a single row of $L$ may have at most $c'w'/2$ many nonzero entries. For a single column of $L$ the same bound holds as these are symplectic products of different generators.

Directly merging this two parts, we obtain an LDI form with weights $w^2(1+(2c-1)/2)$ and register degrees $(2c-1)(1+w^2/2)$. Both of these are $O(w^2c)$. For a bound on the maximal entry, note that $B\leq (q-1)+\max_{i,j}|\phi_\infty(s_i)\odot\phi_\infty(s_j)|=(q-1)+(q-1)^2w'$, or $(q-1)[1+(q-1)w^2]$.
\end{proof}

This improved bound aids in our tighter bound for the local-dimension required to promise the distance is preserved. In the following Theorem we provide a myriad of possible cases as qLDPC does not strictly require $O(1)$ weights but just has low weights, likewise it's possible for a code to be low weight or low register degree but not both, and this bound still provides improvements.

\begin{lemma}
Let $S$ be a CSS stabilizer code with local-dimension $q$ with weights given by $(w_x,c_x,w_z,c_z)$. Then there is an LDI form for $S$, denoted $S_\infty$, such that the maximal entry, $B$, satisfies $B\leq \lfloor q/2\rfloor^2\min(w_x,w_z)$ and the new weights are at most $\max(w_x,w_z)^2\max(c_x,c_z)$ while the new register degree is at most $\max(w_x,w_z)\max(c_x,c_z)$.
\end{lemma}

The specialization to CSS codes will greatly reduce the cutoff value required to ensure the distance. Further, if the initial local-dimension is $2$, then quantum weight reduction can be employed here, providing quite small constants in the expression for the cutoff value. We will make this explicit following our Theorem for qLDPC LDI codes, but first prove this result.

\begin{proof}
We will handle this in two parts. First we will prioritize reduction of $B$, then we will show the associated weights. Let $H_X=[A\ \ B]$ and $H_Z=[C\ \ D]$, where the first block corresponds to the first $k$ columns for each. Without loss of generality, we will consider performing a reduction on the $H_Z$ parity check portion. Assume $H_Z$ has full row rank after removing redundant generators. Before elimination, $H_Z$ may be completely arbitrary. Choose an invertible row-operation matrix $U\in\mathrm{GL}_{k_Z}(\mathbb{Z}_q)$ and a qudit-permutation matrix $P$ such that
\begin{equation}
G_Z:=UH_ZP=[I\mid D].
\end{equation}
These are permissible operations that will not change the codespace. The same column permutation must be applied to $H_X$. Write
$H_XP=[A\mid C]$. Since the original code satisfies $H_XH_Z^T=0$, the transformed matrices still commute:
\begin{equation}
(H_XP)G_Z^T=H_XH_Z^TU^T=0.
\end{equation}
With the block forms above, this is simply $
A+CD^T=0\pmod q.$ This relation tells us exactly what the pivot block should become after lifting into an LDI form.

We now replace all of the entries in each block with a centered version so that we subtract $q$ from entries that are at least $q/2$, returning entries that at most have absolute value $\lfloor q/2\rfloor$. Denote these new blocks with tildes on top. Define
\begin{equation}
H_Z^{\mathrm{LDI}}=[I\mid\widetilde D],\qquad
H_X^{\mathrm{LDI}}=[-\widetilde C\widetilde D^T\mid\widetilde C].
\end{equation}
The replacement of $A$ is legitimate because
$-\widetilde C\widetilde D^T\equiv-CD^T\equiv A\pmod q$, so the code is still equivalent to the original code when reducing by $q$. More importantly,
\begin{equation}
H_X^{\mathrm{LDI}}\bigl(H_Z^{\mathrm{LDI}}\bigr)^T
=-\widetilde C\widetilde D^T+\widetilde C\widetilde D^T=0.
\end{equation}
The final code is still CSS and reduces to an equivalent source code modulo $q$, and the $X$--$Z$ commutation is now exactly zero over $\mathbb{Z}$.

At this point, the only entries that can grow beyond the centered range are the entries of $-\widetilde C\widetilde D^T$. To see what dictates their maximal value, define
\begin{equation}
\omega=\max_{i,j}\left|\text{supp}(C_{i,:})\cap\text{supp}(D_{j,:})\right|.
\end{equation}
Each nonzero product contributing to an entry has magnitude at most $\lfloor q/2\rfloor^2$. Therefore
$B\le\lfloor q/2\rfloor^2\omega$.

This bound is useful because $\omega$ can be read as an overlap count. While, with the choice made to put $H_Z$ into canonical form, the weights may have been increased the support of the $H_X$ checks have not been altered. In turn, this means that $\omega\leq w_x$. If instead $w_z$ is lower, we can instead perform a canonical reduction on $H_X$, so $B\leq \lfloor q/2\rfloor^2\min(w_x,w_z)$ is possible.

We now address the potential weight increases from this process. Assuming that $H_Z$ is put into canonical form, $H_X$ remains unchanged so $w_x,c_x$ are also unchanged. To put $H_Z$ into canonical form, we are simply performing traditional sparse Gaussian elimination into reduced-row echelon form. In this process each row addition to zero out a column may increase the recipient row's weight by $w_z$, this must be done in total up to $w_zc_z$ many times, increasing the row weight to at most $w_z^2c_z$. The register degree can be increased by one each time, thus possibly increasing the register degree to $w_zc_z$.
\end{proof}

For $q=2$ or $3$, where $\lfloor q/2\rfloor=1$, this reduces to $B\le\min(w_X,w_Z)$. There is an important caveat to this prior result. The bound on $B$ and the sparsity after elimination are not the same issue as each component will factor into the final cutoff bound. Fill-in in the systematicized matrix may increase check weights and register degrees even when the $B$ bound remains small. Therefore both orientations should be compared when preserving the qLDPC structure is important, not only the final value of $B$. This all said, if one begins with a qubit code, $q=2$, and apply quantum weight reduction methods a code with parameters $(5,3,5,5)$ is obtained \cite{hastings2021quantum,sabo2024weight,hsieh2025simplified}. This would in turn provide $B\leq 5$ and maximal weight of up to $75$. Due to subtle aspects of the following Theorem, we will obtain a $p^*$ bound of $\min(25^{\min(d-1,\bar{k_z})},15^{\min(d-1,\bar{k_x})})$ with $\bar{k_z}$ being the number of $Z$-type generators and $\bar{k_x}$ being the number of $X$-type generators \footnote{$25$ is due to $B\leq 5$ and there being up to $5$ nonzero entries per row in $H_Z$, while $15$ comes from the new register degree for $H_X$.}.

We now turn to our cutoff bounds for qLDPC LDI codes. This result utilizes the prior results to reduce prior expressions to more manageable counts.

\begin{theorem}\label{mainthm}
Let $S$ be a qLDPC CSS stabilizer code with parameters $[[n,k,d]]_q$, for some prime $q$. Let $S_\infty$ be the local-dimension-invariant CSS code obtained from $S$. Further, let $S_\infty$ have weights $w_x,w_z$ and qubit degrees $c_x,c_z$. Let $w=\max(w_x,w_z)$, $c=\max(c_x,c_z)$, and $\mathcal{W}=\min(w,c)$. Further, let $B$ be the maximal, in absolute value, entry in $S_\infty$. Then $S_\infty$ will have distance $d'\geq d$ for any prime $p$ larger than $p^*$ with $p^*$ satisfying:
\begin{itemize}
    \item For $\mathcal{W}=\Omega(d^{1/2})$:
    \begin{equation}
        B^{d-1}\mathcal{W}\left( \mathcal{W}-\frac{\mathcal{W}( \mathcal{W}-1)}{d-2}\right) ^{(d-2)/2}
    \end{equation}
    \item For $\mathcal{W}=O(1)$:
    \begin{equation}
        B^{d-1}2^{(d-1)(\mathcal{W}-1)/3}
    \end{equation}
    \item For $\mathcal{W}<d$, $p^*\leq W^{d-1}$ with:
    \begin{equation}
W=\max_{\text{X,Z checks}}\min(\max_{\text{rows } r}\sum_{i:|i|\leq d-1} \lvert m_{ri} \lvert,\max_{\text{columns }c}\sum_{i:|i|\leq d-1}\lvert m_{ic}\lvert )\leq \mathcal{W}B.
\end{equation} 
\end{itemize}
\end{theorem}

This Theorem covers the case where the code is not $O(1)$ sparse, where it uses a slightly tighter bound than priorly used, while for $O(1)$ weights a far tighter bound is possible. Beyond that, this result allows for flexibility in the LDI generation procedure. Depending on the particular code sometimes the weights being optimized is the best option, while other times the maximal entry ought to be optimized. Notably, this result can leverage sparsity in either the weights of the checks or on the register degree, so a code that is large in one of these parameters can still use a tighter bound.

All of the bounds are based on determinant bounds. For some of these we directly refer to the appropriate literature result for the bound, utilizing it once the problem has been reduced.

\begin{theorem}[Ryser's Theorem \cite{ryser1956maximal}]
    For an $n\times n$ binary matrix with $kn$ many ones, the determinant is at most:
    \begin{equation}
        k\left( k-\frac{k(k-1)}{n-1}\right)^{(n-1)/2}
    \end{equation}
\end{theorem}
While this bound grows slower than Hadamard's inequality, the difference can be quite small except for once $k$ becomes on par with $\sqrt{n}$, meaning a total of $n^{3/2}$ total nonzero entries. Newer and tighter bounds have since been shown, currently the best result is as follows:

\begin{theorem}[ABW Bound\cite{araujo2022maximum}]
Every matrix $A\in \{0,1\}^{n\times n}$  containing at most $n+k$ non-zero entries has determinant at most $2^{k/3}$.
\end{theorem}

It is important to note here that if $k$ grows faster than $n$, and thus a non-$O(1)$ qLDPC code, then this bound is worse than Hadamard's inequality. Precisely for this reason, our Theorem has numerous cases.

\if{false}The following Theorem is from \cite{scheinerman2019maximal}:
\begin{theorem}\label{sparsehadamard}
    Let $k\geq 2$ be an integer, and $M$ be an $n\times n$ matrix with at most $kn$ many $1$ entries in each row, \textcolor{teal}{SI: Maybe its $kn$ ones "in total", not in each row as everyone will satisfy it.} \textcolor{red}{LG: yes, I've apparently chosen the wrong result. There were a few different bounds shown and this one was the loosest in some sense.} with the rest $0$. Then $det(M)\leq c_k^n$, with $c_k=(\sqrt{k^2-1})^{(1-1/k)/2}k^{1/(2k)}$.
\end{theorem}
This will act as our central tool in showing equation \ref{sparsehbound}.\fi

We proceed to show our new Theorem now. To show our new result, we will need the following pair of nuanced definitions from \cite{gunderman2024stabilizer}:

\begin{definition}
An \textbf{unavoidable error} is an error that commutes with all stabilizers and produces the $\vec{0}$ syndrome over the integers.
\end{definition}

These are errors which are undetectable regardless of the local-dimension of the system, hence are unavoidable even when considering the code over a different local-dimension. The other possible case are undetectable errors that are only undetectable due to the selection of the local-dimension as the syndrome is only congruent to $0$, but not $0$ over the integers.

\begin{definition}
An \textbf{artifact error} is an error that commutes with all stabilizers but produces at least one syndrome that is only zero modulo the base.
\end{definition}

Importantly, the set of unavoidable errors does not change as the local-dimension does. These correspond to matrix minors which have a determinant of truly $0$, while those errors that correspond to artifact errors being introduced may occur when the determinant of a matrix minor is a multiple of the local-dimension, thus resulting in a new vector in the kernel which might have lower Pauli weight that the distance of the original code. With this insight, we then aim to upper bound the determinant, so that the local-dimension is larger than the largest possible determinant for a matrix minor, in which case we know that no new artifact errors can be introduced which could lower the distance.

The core proof, reducing the problem to determinant bounds, follows the argument from the prior works on LDI codes.\cite{gunderman2024stabilizer}
\begin{proof}[Proof of Theorem \ref{mainthm}]
To begin, we relate this problem to bounding determinants of matrix minors within the LDI stabilizer code. Let the integer LDI stabilizer matrix be \(H=[H_X\mid H_Z]\), and an error is represented by the row vector \(e=(e_X\mid e_Z)\). The syndrome row vector is given by the symplectic product between the stabilizer generators and the error, namely \(e_ZH_X^T-e_XH_Z^T\). Therefore, defining \(A_{\infty}=[-H_Z\mid H_X]\), we can write the integer syndrome simply as \(eA_{\infty}^T\).
An error vector \(e\) is undetectable in prime local dimension \(p\) when \(eA_{\infty}^T=0\pmod p\). However, zero syndrome does not necessarily imply a logical error. Let \(S_p=\operatorname{rowspan}_{\mathbb F_p}(H\bmod p)\) denote the stabilizer space. Write \(K_p=\{e\in\mathbb F_p^{2n}:eA_{\infty}^T=0\}\) for the space of zero-syndrome row vectors. Since the generators commute, \(S_p\subseteq K_p\). A zero-syndrome row vector represents a stabilizer if it lies in \(S_p\), and a nontrivial logical operator otherwise, with overall phases ignored. Consequently, \[d_p=\min\{\operatorname{wt}(e):e\in K_p\setminus S_p\},\] where \(\operatorname{wt}(e)\) counts the registers on which the error acts nontrivially.\begin{itemize}
    \item If \(eA_{\infty}^T=0\) already over the integers, then the error is undetectable independently of the local dimension and is therefore unavoidable.
    \item  If instead \(eA_{\infty}^T\neq0\) over the integers but \(eA_{\infty}^T=0\pmod p\), then every nonzero syndrome entry is a multiple of \(p\). The error has therefore become undetectable only after reduction modulo \(p\), which is the artifact-error case.
\end{itemize}

Fix a set \(T\) containing \(j<d\) registers. Let \(A_T\) be the part of \(A_\infty\) containing the \(X\)- and \(Z\)-columns associated with the registers in \(T\). Then an error \(e_T\) supported only on \(T\) has syndrome \(e_TA_T^T\). Since \(A_T\) has \(2j\) columns, the square minors considered below have order at most \(2j\le 2(d-1)\). Let
$t=\operatorname{rank}_{\mathbb Q}A_T.$ If \(t>0\), then \(A_T\) contains a \(t\times t\) minor \(M\) with \(\det M\neq0\). If \(p\) does not divide \(\det M\), then this determinant remains nonzero modulo \(p\), so \(M\) is still invertible over \(\mathbb F_p\). Thus \(A_T\) still contains a rank-\(t\) minor modulo \(p\). Since reduction modulo \(p\) cannot increase the rank of an integer matrix, we also have \(\operatorname{rank}_{\mathbb F_p}A_T\le t\). Therefore
$
\operatorname{rank}_{\mathbb F_p}A_T
=
\operatorname{rank}_{\mathbb Q}A_T=t.
$\\

Suppose every such chosen nonzero minor satisfies \(|\det M|\le p^*\). Then any prime \(p>p^*\) cannot divide \(\det M\), so the above rank equality holds for every support \(T\) with \(|T|<d\). Thus reduction modulo \(p\) does not create any new zero-syndrome values on these supports.\\
Let \(F\) denote either \(\mathbb Q\) or \(\mathbb F_p\). Let \(E_T(F)\) be the space of errors supported only on \(T\), let
$
S_F=\operatorname{rowspan}_F H $ be the stabilizer space, and define
\[
K_T(F)=\{e\in E_T(F):eA_\infty^T=0\}.
\]
Thus \(K_T(F)\) is the space of zero-syndrome errors supported on \(T\). When \(F=\mathbb F_p\), all matrices are understood modulo \(p\). The logical operators supported on \(T\) are therefore described by 
\[
L_T(F)
=
\frac{K_T(F)}
{S_F\cap E_T(F)}.
\]

Let \(r_F=\operatorname{rank}_F H\) be the rank of the full stabilizer matrix and \(H_{T^c}\) be the part of \(H\) outside the registers in \(T\). Then
$
(r_F-\operatorname{rank}_F H_{T^c})$ gives the number of independent stabilizers supported entirely on \(T\). Finally the number of independent logical operators supported on \(T\) is given by 
\[
\begin{aligned}
\ell_T(F)
&=\dim(\text{zero-syndrome space})
-
\dim(\text{stabilizer subspace on }T)
 \\
&=
2|T|-\operatorname{rank}_F A_T- (r_F -\operatorname{rank}_F H_{T^c})\\
&=
2|T|-\operatorname{rank}_F A_T- r_F +\operatorname{rank}_F H_{T^c}
\end{aligned}
\]
If     $\ell_T(F)=0,$ then every zero-syndrome operator on \(T\) is a stabilizer.
So there is no logical error supported on \(T\).
If this is true for every $|T|<d,$ then there is no logical operator of weight below \(d\), hence $\boxed{d_F\ge d}.$
 A key takeaway is that the preservation of distance reduces to a determinant bounding problem.

Next, specializing to CSS codes reduces the size of the matrices whose determinants we must consider from being up to $2(d-1)$ to only $d-1$. For the first two bounds, we apply Ryser's Theorem and ABW's Bound, respectively. 

For our other bound, let $\lambda_1$ be the largest, in absolute value, eigenvalue for a $d\times d$ matrix $M$. $M$'s largest eigenvalue then satisfies $|\lambda_1|\leq \max_r \sum_{i=1}^d |M_{ri}|$. The bound is then simply using $|det(M)|=|\prod_i \lambda_i|\leq |\lambda_1|^{d-1}$. We may be able to provide a better bound on $\lambda_1$, by noting that $det(B)=det(B^T)$, thus we can also say $\lambda_1\leq \max_c \sum_{i=1}^d |M_{ic}|$. We may then use the smaller of these two, but must use the larger of the two values we obtain from the $X$ and $Z$ checks. This provides the expression for $W$. Importantly this means that if the code is extremely low weight but not register degree, we can still obtain a small bound, and vice versa.
\end{proof}

\begin{example}
As an example of the utility of this result, if we apply classical weight reduction to a LDPC input parity check matrix of the lifted-product construction \cite{panteleev2021quantum,sabo2024weight} results in a CSS qLDPC code with register degree at most $3$ within each of the $X$ and $Z$ checks. Assuming we begin with a binary classical code, lift the code and alter the coefficient of one block, then the last form of our bound provides $p^*= 3^{d-1}$, while the second form provides $p^*=2^{2(d-1)/3}$. This greatly improves the distance preservation bound for another construction of quantum error-correcting codes, in this case where the possible parameters go as $[[n,\Omega(n^{1-\alpha/2}/\log n),\Omega(n^\alpha \log n)]]$ with $0\leq \alpha<1$.
\end{example}

\section{LDI Topologic Codes}\label{topo}

Topologic quantum codes, with bounded stabilizer radius, are popular examples of qLDPC codes. We examine some cases of LDI topologic codes. There have been prior generators for the toric code which have been LDI \cite{xu2024letting,gunderman2024stabilizer}. However, we provide here more general plaquettes for LDI toric codes. These generators are shown on the right in Figure \ref{plaquettess}. Prior LDI generators have only utilized $\pm 1$ powers for the operators. For a field it is possible to relabel the generators so that a pair of $\{a,b\}$ and $\{c,d\}$ are both $1$, which results in the remaining powers either being additive inverses or factors of $-1$ in the field. Notably, unlike the $\pm 1$ generators, the codespace will no longer appear as purely an extension of the qubit toric code, but rather as biased spaced out lattice supports within the hypercube. If these generators are used to generate a GKP-like stabilizer code, the codewords will be distorted compared to the $\pm 1$ generated LDI toric code. Notably, these generators form the same topologic space and so the distance and logical register count is unchanged.

The square toric code is able to be generated by a pair of translationally invariant generators. Remaining in Euclidean space, the other lattice is the dual hexagonal/triangular lattice. An LDI version of this is not able to be generated with only a pair of translation invariant generators, instead a trio of generators are required, where there are a pair of additive inverse vertex generators, so that the topologic expanding property is preserved. The generators for a general LDI hexagonal/triangular lattice are shown on the left in Figure \ref{plaquettess}. While there are three equations to satisfy, one of them is redundant. The additional flexibility in the parameters permits for more ways that the codeword space can be selected. The parameters are again preserved due to topologic arguments.

\begin{figure}[!t]
    \centering
    \includegraphics[width=\linewidth]{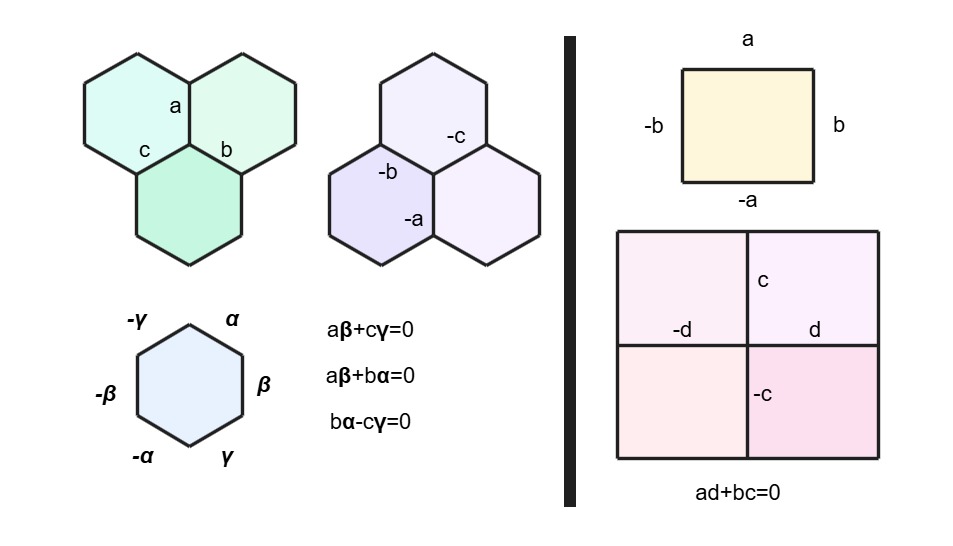}
    \caption{Fully general LDI plaquettes for toric codes (right) and hexagonal/triangular surface codes (left). These suffice to generate the code for any integral-domain, while for other rings more generators may be employed.}
    \label{plaquettess}
\end{figure}

Permitting ourselves to move away from Euclidean space, we can also consider quantum codes constructed on Euclidean closures of hyperbolic lattices. These codes are attractive due to having qubit parameters $[[n,O(n),O(\log n)]]$, so being able to transform them into LDI forms would permit low weight, geometric codes with a constant rate. A 2D hyperbolic lattice is specified by a pair of Schläfli symbols $\{m,n\}$ with $1/m+1/n< 1/2$ where they indicate that the lattice has $m$ $n$-gons meeting at each vertex. For these we make use of our sign swap algorithm Theorem \ref{sswap} to create codes with small valued maximal entries.

\begin{lemma}\label{sshyper}
    Any Surface code with Schläfli symbol of form $\{m,n\}$ where at least one of $m,n$ is even can be made LDI with the sign swap algorithm, and thus has maximal absolute valued entry $B=1$.
\end{lemma}

This gives an LDI form for these codes. While this does not satisfy the boundary expansion property and so the distance is not immediately preserved, this has $B=1$ and register degrees of $2$. So a $p^*= 2^{(d-1)/3}$ is obtained from the second bound, which given hyperbolic codes have $d=O(\log n)$ for $n$ physical registers, makes this a comparatively slowly growing set of local-dimensions which must have their distance verified through other means. In total the number of local-dimensions that must be checked is polynomial in $n$.

\begin{proof}
    To prove this, we must show that either the $X$ or the $Z$ checks form a bipartite conflict graph $G_C$ and that $wt(s_i) \ge |G_C(s_i)| +1$ holds. The former is a property of the lattice given even $m$ or $n$. For a $\{2m',n\}$ surface code, the constraint graph $G_C = G_Z$ is a cycle graph of length $2m'$, which is $2$-colorable. Similarly, $\{m, 2n'\}$ leads to a constraint graph that is a cycle graph of length $2n'$.

    For a given $\{m,n\}$ surface code, we see that the $X$ check weight is $n$, and the $Z$ check weight is $m$. Choosing an even cycled stabilizer set arbitrarily picked to be $X$, we see that check weight $wt(s_i)$ is $n$. The conflict graph for a single stabilizer on a vertex $G_C$ is an even cycle of length $n$. Alternating signs of the powers of the Pauli operators on the qubits in the support will then have precisely two overlapping supports, with alternating sign, for each dual stabilizer--or none at all.
\end{proof}

This handles three quarters of all hyperbolic surface codes, however, we also aim to have a boundary expanding set of generators so that the distance is automatically preserved. Unfortunately since the generators are not translated, but rather rotated and translated to form the tiling, this makes designing boundary expanding generators tricky. For instance, simply applying rotations on an alternating $n$-gon, with $n$ even, a conflict occurs making a direct rotation a non-boundary expanding set of generators. Given that, the slowly scaling $p^*$ cutoff value reduces the problem space.

\begin{table}
    \centering
    \begin{tabular}{||c|c|c||}
    \hline
    Source & Bound & Requirements \\ \hline\hline
        \cite{gunderman2022degenerate} & $[B(q-1)(d-1)[1+(d-1)^2(q-1)^{d-1}(d-2)^{(d-2)/2}]]^{d-1}$ & None   \\
        \cite{gunderman2020local} & $B^{2(d-1)}[2(d-1)]^{d-1}$ & None \\
         \cite{moorthy2023local} & $B^{d-1}(d-1)^{(d-1)/2}$ & CSS  \\
        Theorem \ref{mainthm} & $B^{d-1}2^{(d-1)(\mathcal{W}-1)/3}$ & CSS, Weight $\mathcal{W}$ is $O(1)$   \\
         Theorem \ref{mainthm} & $\max_{\text{X,Z checks}}\min(\max_{\text{rows } r}\sum_{i:|i|\leq d-1} \lvert m_{ri}\lvert,\max_{\text{columns }c}\sum_{i:|i|\leq d-1}\lvert m_{ic}\lvert )\leq (\mathcal{W}B)^{d-1}$  & CSS, Weight $\mathcal{W}<d$ \\
         \hline
    \end{tabular}
    \caption{Comparison of various bounds for $p^*$. The last bound is of particular use when few entries are altered so that the $B$ contribution is managed.}
\end{table}



\section{Conclusion}\label{conclusion}

Local-dimension-invariant stabilizer codes allow quantum information to be protected with additive codes regardless of the local structure of the registers. As all codes can be put into an LDI form, various desirable parameters of these codes can be preserved. Unfortunately, the distance of LDI forms is not immediately carried over unless the local-dimension is sufficiently large. In this work we centered this problem and reduced the largeness of the sufficient local-dimension by specializing to quantum low-density parity-check stabilizer codes. By restricting the discussion primarily to these codes, we are able to apply tighter bounds on the distance preservation condition. Notably the determinant bound term can be reduced from an exponential function with a base that also scales with the distance as well as the number of logicals, to only an exponential function of the distance. In practice, particularly when paired with a sign-swap LDI form construction, these bounds become manageable values.

Beyond providing new bounds on sufficient local-dimension for the distance preservation condition to hold, we also analyze the implications of the distance preservation bound itself. It can be interpreted as a squeezing requirement for GKP-like encodings or as a precision requirement for CV encodings. We also illuminate new degrees of freedom for two-dimensional surface codes, allowing for bias in the codeword space. We also show a non-boundary expanding construction for hyperbolic tilings.

The results of this paper improve the practicality of the LDI extension of the stabilizer formalism, as well as expand the implications of this line of inquiry. One direction that seems worth exploring would be a simpler sufficiency condition than exhaustive enumeration in Theorem \ref{swapassignment}. This would likely reduce to some graph theoretic constraints, however, those are to be found. A seemingly worthwhile future direction could be to examine the possible transversal logical structure for these codes. This might be easiest to interpret as coarse-graining of transversal gates for continuous-variable encodings, however, even more exotic features may exist. We leave those as future work.

\if{false}
\clearpage

Stuff to still be incorporated or formally cut.

\if{false}
Improved bounds on sparsity after LDI transform (what I have already): Before doing any reductions, note that due to row and column weight being $w$, each generator can only have overlap with at most $w^2$ generators, and when putting codes into an LDI form you only need to alter 1 entry per non perfectly commuting generator. So there ought to be a $w^2$ form or perhaps better. That being said, for a rigorous proof I haven't hit that from mental math. In performing Gaussian elimination the canonical form will have weights $O(w^2)$, from which the prior argument gives $O(w^4)$. Still "O(1)" but awfully big.

I think I can keep the weight as $O(wc^2)$ although the $B$ value goes from $6$ to $6c^2$ with row weight $w$ and column weight $c$. While this drastically boosts $B$ it may be worth the tradeoff for the lower weight--I suspect we'll have a couple of variable bounds where they'll cross over for sufficiently dense codes or something.\fi

\if{false}\begin{proof}[working draft of proof]
Let us begin by supposing that we were to perform Gaussian elimination on $S$ to put it into canonical form $[I\ X_2\ |\ Z_1\ Z_2]$. In the process of performing Gaussian elimination on the $k$ pivots, we may need to add $w-1$ rows to each row representing a pivot, therefore increasing the row weights to $w(w-1)$. Each generator in the canonical form will intersect with a most $[w(w-1)](c+w)$ other generators, whereby a single entry in the $L$ matrix for each pair suffices for setting the symplectic product to precisely $0$.
\end{proof}\fi

\if{false}
\section{CSS Codes}

\if{false}
\begin{theorem}
    The minimum check weight required to convert a binary or even local dimensional CSS code into LDI form with only sign swaps is $wt(s) \ge |G_C|_{s} +1 ~ \forall s \in {H_X}$ given that the conflict graph $|G_C|$ is bipartite.
\end{theorem}

\begin{proof}
    We can arbitrarily pick $H_X$ or $H_Z$ to be signed. For this proof, we pick $H_X$. We see that for a stabilizer row $s$ we can sign any qubit independently, but since global sign flips can be factored out we have $wt(s)-1$ degrees of freedom. We must now define the Conflict Graph $G_C$ for a given code to sign our stabilizer matrix.

    We construct the conflict graph $G_C$ by creating nodes at every $1$ on $H_X$. We connect two nodes internally in $s_i$ if they both overlap with $s_j$. If this graph is bipartite, a $B=1$ LDI form of the code exists with only sign swaps.

    With BFS, we can $2$-color this graph with $1,-1$ to sign $H_X$ and yield an LDI stabilizer matrix. A proof of commutation follows from the fact that any pair of qubits that would commute to $1\cdot1 + 1\cdot1 = 0 ~mod~2$ now commute to $1\cdot-1 + 1\cdot1 = 0$, leading to a total stabilizer overlap of $\sum(1\cdot-1 + 1\cdot1) = 0$

    If $H_X$ is not $2$-colorable, we can pick $H_Z$ to sign instead. However, if neither conflict graph is $2$-colorable, we must set up $\pm2$ valued nodes to reconcile odd cycle commutativity. (to be done but I lean towards this being possible; edit: not anymore, finding solutions to this requires np-hard ILP solvers--I think I agree there)

    Each sign flip consumes one degree of freedom. Therefore for a given code, a stabilizer row can be put into LDI form with just sign flips if
    \begin{equation}
    wt(s_i) \ge Rank(G_C)_{s_i} +1
    \end{equation}
    where $|G_C|_{s_i}$ is the number of stabilizer rows of $H_Z$ $s_i$ has a unique overlap with.
\end{proof}
\fi

\textcolor{red}{LG: end of rewrite attempt. I kept trying to figure out a way to go beyond qubits but so far haven't cracked it. For the below Corollary I believe I agree that it should hold for Clifford, but debating if there's any surprises for beyond Cliffords}

\begin{corollary}[Corollary to \ref{sswap}]
    Any transversal gates present in the original code are preserved with this transformation that does not change the connectivity of nodes of the stabilizer matrix. We see that for any stabilizer row $s_{LDI}$ where connectivity is identical, entry $q_{i_{LDI}}$ = $q_i\cdot b$ where $b$ is the value of the entry. We can then extend the standard argument where we take $U =\bigotimes_{i=1}^n U_i$ to show that if the original code has a transversal gate, the LDI transformed code does too. \begin{align}
    U_iq_{i_{LDI}} U_i^\dagger & = U_i(b \cdot q_i)U_i^\dagger \nonumber \\
        & = b\cdot U_iq_i U_i^\dagger  \nonumber \\
        & = b\cdot q_i'
    \end{align}
    where $q_i'$ is the output of the original transversal gate.

    Probably true for DFT, sum gate is automatically carried over. 
\end{corollary}

\fi

\section{Geometric locality}

\begin{theorem}
    The $p^*$ bound of a sign swap derived $B=1$ code can be defined by the odd cycle packing number $v$ of the unsigned matrix lattice, resulting in $p^* = p > 2^v$. 
\end{theorem}
\begin{proof}
    The signed half of the code is TU due to being equivalent to the matrix representation of a directed graph and can be ignored. We must now show that each odd cycle in a code with $B = 1$ leads to a square sub-matrix with determinant of 2 in unweighted cases. The incidence sub-matrix of a given cycle is given by a circulant matrix with the diagonal being the identity matrix. For a cycle size of l, this generalizes up to choice of sign to the $l\times l$ matrix \begin{equation}
        M_l = \begin{pmatrix}
        1 & 0 & 0 & \dots & 1 \\
        1 & 1 & 0 & \dots & 0 \\
        0 & 1 & 1 & \dots & 0 \\
        \vdots & \vdots & \vdots & \ddots & \vdots \\
        0 & 0 & \dots & 1 & 1
        \end{pmatrix}
    \end{equation}
    for any $l$-cycle. Expanding the determinant across the first row gives 
    \begin{equation}
        1 \cdot det(L) + (-1)^{l-1}\cdot det(R) = 1 - (-1)^l 
    \end{equation}
    This expression is 0 for even cycles but 2 for odd cycles. Thus, the maximal determinant is one which includes all geometric frustrations caused by odd cycles, leading to $p^* = 2^{n_{odd}}$ where $n_{odd}$ is the number of disjoint odd cycles. This is quantified by the odd cycle packing number of the lattice $v$, leading to the bound outlined above.
\end{proof}



\if{false}\section{Transversal operations on certain LDI codes}

If a code has transversal $H$ then the LDI form can have transversal $H_\infty$, I believe (just apply the same transformation). Will the same hold for transversal DFT?

Can we say anything about generalized $T$ gates and keep transversality? 
Transversal $T$ gates that are built on top of triorthogonal codes stay triorthogonal with purely signed LDI construction (caveat: if possible, not sure about the L construction), and thus $T$ gates maintain transversality. 

\fi

\section{Additional interpretation of $p^*$}

As a precision requirement on CV implementations. Basically, we obviously don't have $\mathbb{Z}$ (or $\mathbb{R}$) accuracy. What level is ok actually? Victor Albert, Christophe Vuillot, Yijia Xu (Tiger codes), Vuillot Terhal Ciani (Torsion rotor codes)

\section*{Archive}
These are theorems that are probably obsolete but could maybe be useful if we say anything about combinatoric $p^*$.

Couldn't really figure out how to translate this to signed LDI.

Totally unimodular codes over $\mathbb{F}_2$ are not optimal but have good decoding properties. Not sure if worth including however.

\begin{definition}
    We define an LDI code to be universally LDI if it maintains its distance for any choice of basis, or alternatively its $p^*$ value is 2.
\end{definition}

\begin{definition}
    A Matrix is Totally Unimodular (TU) if the determinant of every square minor of the matrix is in the set $\{-1,0,1\}$
\end{definition}
By Definition 14 from \cite{gunderman2022degenerate}, we see that if the minor determinant is 1 over the integers, no null space exists in any local dimension, making our code $p^* = 2$. This can done by the construction of a Totally Unimodular (TU) code. This is overcomplete however, since syndromes may only partially reduce $mod~p$ for a given error while maintaining detectability.


\begin{theorem}
    CSS codes are universally LDI if the X and Z checks are TU.
\end{theorem}
\begin{proof}
    LDI formulation preserves CSS structure. Furthermore, CSS codes take the structure:
    \begin{equation}
        \begin{aligned} \begin{bmatrix} {\mathcal {X}} &{} | &{} 0\\ 0 &{} | &{} {\mathcal {Z}} \end{bmatrix}. \end{aligned}
    \end{equation}
    The maximal minor determinant in this case is given by the product of minor determinants $max(det_{minor}(\mathcal{X})) \times max(det_{minor}(\mathcal{Z}))$ which is still contained in the set $\{-1,0,1\}$.
    In case either $\mathcal{X}$ or $\mathcal{Z}$ are not TU, the maximal minor determinant has absolute value greater than 1, leading to artifact errors. Note that this is overcomplete.
\end{proof}
This may be useful for the decoding of codes over multiple bases since the syndrome space is unchanged no matter the local dimension.

\if{false}
The $p^*$ bound is currently defined by a combinatorial argument that bounds the worst case maximal syndrome element that can be generated by a given code. This is over-complete however, since for the syndrome module over $\mathbb{Z}$, $\mathcal{S}_\mathbb{Z}$, the isomorphism $S_\mathbb{Z} \cong S_p$ guarantees that code distance is maintained.

\begin{theorem}
    For the syndrome module $\mathcal{S}_\mathbb{Z}$ generated by a code, any prime basis $p$ such that $p$ does not belong to the Potential Torsion group $\mathcal{T}$ of the code does not drop code distance.
\end{theorem}


\textcolor{blue}{TODO before assembly: Rewrite this to show it actually proves no new degeneracy is introduced, couple with monte carlo GCD checking theory which forces gcd(synd) > gcd(err) for artifact errors. PROBABLY SCRAP, CAN'T FIND A USE FOR IT}

\begin{proof}
    Since the LDI stabilizer matrix is defined over $Z$ which is a Principal Ideal Domain, we can decompose the syndrome space $\mathcal{S} = Coker(S^T)$ into its Smith Normal Form (SNF) $\mathcal{S} = UDV$. Since $U,V \in GL(\mathbb{Z})$ and are unimodular, and $D = diag(\sigma_1, \sigma_2,\dots, \sigma_r, 0, \dots,0)$ the syndrome module over $\mathbb{Z}$ is 
    \begin{equation}
        \mathcal{S}_\mathbb{Z} \cong  \bigoplus_{i=1}^{r}(\mathbb{Z/}\sigma_i \mathbb{Z})\oplus \mathbb{Z}^k 
    \end{equation}
    Upon reduction modulo some $p$, we get \begin{equation}
        \mathcal{S}_p \cong \mathcal{S}_{\mathbb{Z}}\otimes \mathbb{F}_p
    \end{equation}
    The fact that $U,V$ are unimodular and non-unique allows us to change basis to $\mathbb{F}_p$ without destroying $U$ or $V$. By the Universal Coefficient Theorem for homology, the rank of $\mathcal{S}$ modulo $p$ is only preserved iff $p$ does not divide any non-zero invariant factor $\sigma_i$. We must also avoid destroying any connections in the stabilizer matrix and therefore need to avoid all primes that divide any entry in the stabilizer matrix. We define the Potential Torsion group $\mathcal{T}$ to be
     \begin{equation}
        \mathcal{T} = \{p : p \ |\ \sigma_i ~\exists i ~or~ p|q_i \in S~\exists i\}
    \end{equation} If chosen basis does not belong to $\mathcal{T}$, $rank(\mathcal{S}_\mathbb{Z}) = rank(\mathcal{S}_p)$. By the rank-nullity theorem, this implies $\dim(\ker(\mathcal{S}_p)) = \dim(\ker(\mathcal{S}_\mathbb{Z}))$. By the Universal Coefficient Theorem, we can conclude $\ker(\mathcal{S}_\mathbb{Z}) \otimes \mathbb{F}_p \cong \ker(\mathcal{S}_p)$.
    This implies that for a modular reduction by a prime not in the torsion set, there are no new low-weight logical operators that arise purely from torsion effects and the homology structure of the syndrome space is preserved, thus $d_{p} \geq d \ \forall p \notin \mathcal{T}$. The injectivity of the map on the generated sub-modules guarantees that the distinguishability of error syndromes is preserved from $\mathbb{Z}$ to $\mathbb{F}_p$. 
    With this, we can define $p^*_{min}$ to be the smallest prime outside the potential torsion group of a code.\begin{equation}
        p_{min}^* = p > p_i \ \forall p_i \in \mathcal{T}
    \end{equation}
    We may compute this over a basis of our syndrome space since SNF invariant factors are independent of choice of basis, lower bounding $p^*$ for a given code.
\end{proof}
By Chinese Remainder Theorem, we can extend this to composite rings of form $Q = \prod^n_{i=0} p_i^1 ~\forall p_i \notin \mathcal{T}$ following the arguments in \cite{gunderman2024stabilizer}.

\fi

\begin{remark}
    It is important to note that operating outside the "Potential Torsion group" here does not lose qubits that were originally encoded by torsion. In the codespace, LDI transformations simply move orignal dimension d-torsion up to q-torsion since $\bar{X}_q\bar{Z}_q = e^{2\pi i/q}\bar{Z}_q\bar{X}_q$. 
\end{remark}

Runtime of SNF is $O(N^3)$ and sparse SNF is $O(NNZ)$
\if{false}
\begin{theorem}[work in progress]
Let $S$ be a non-degenerate CSS code with parameters $[[n,k,d]]_p$. Then there exists an LDI form such that $d'\geq d$ for all primes at least as large as $p$ (or $B$??).
\end{theorem}

\begin{proof}
WLOG, let there be at least as many $X$ type checks as $Z$ type checks, further let the code be put into a partial canonical form $H_x=[I\ \ X_2]$ and $H_z=[Z_1\ \ Z_2]$. We now construct an LDI form for this code. To the $Z$ checks, add the lower triangular matrix $L$ with entries $L_{ij}=-\phi_\infty(H_{x,i})\odot\phi_\infty(H_{z,j})$ to the first block. The resulting code $H_x=[I\ \ X_2]$ and $H_z=[Z_1+L\ \ Z_2]$ is a valid LDI form:
\begin{eqnarray}
    \phi_\infty(H_x')\odot\phi_\infty(H_z')&=& I(Z_1+L)+X_2Z_2\\
    &=& L+IZ_1+X_2Z_2\\
    &=& L+\phi_\infty(H_x)\odot\phi_\infty(H_z)=0.
\end{eqnarray}
The CSS form permits us to interpret $H_x'$ and $H_z'$ as a pair of classical error correcting codes. Further, since the original code was non-degenerate we may set $d_x$ as the distance of $H_x$ and $d_z$ as the distance of $H_z$ with the stabilizer distance $d=\min(d_x,d_z)$. As $H_x'=H_x$ we may form the codeword space which will have the support unaltered. For $H_z'$ it will be useful to note that $L\mod p=0$. Then using the canonical form with $H_z$... [I need to expand the canonical forms for this to be easier]
\end{proof}
\fi

\subsection{A characterization of local-dimensions that may lower distance}

It is an outstanding question whether the distance of any code is preserved for any prime at least as large as the initial local-dimension upon using a local-dimension-invariant extension of the code. For general stabilizer codes, the required local-dimension can be quite large, and so we describe a way to reduce the set of primes that must have their distance computationally verified. For this we introduce a \textit{potential torsion group} for a given LDI form. Primes that are outside this potential torsion group will preserve the distance of the code, while those inside may have the same distance with the given form, or may not--in which case another LDI form may preserve the distance. The potential torsion group is obtained by considering the Smith Normal Form (SNF) for subsets of pairs of registers, which in turn provides guarantees on the rank of minors not being decreased. This is expressed in the following Theorem:

"""

The $p^*$ bound is currently defined by a combinatorial argument that bounds the worst case maximal syndrome element that can be generated by a given code. This is over-complete however, since for the syndrome module over $\mathbb{Z}$, $\mathcal{S}_\mathbb{Z}$, the isomorphism $S_\mathbb{Z} \cong S_p$ guarantees that code distance is maintained.

\begin{theorem}
    For the syndrome module $\mathcal{S}_\mathbb{Z}$ generated by a code, any prime basis $p$ such that $p$ does not belong to the Potential Torsion group $\mathcal{T}$ of the code does not drop code distance.
\end{theorem}



\begin{proof}
    Since the LDI stabilizer matrix is defined over $Z$ which is a Principal Ideal Domain, we can decompose the syndrome space $\mathcal{S} = Coker(S^T)$ into its Smith Normal Form (SNF) $\mathcal{S} = UDV$. Since $U,V \in GL(\mathbb{Z})$ and are unimodular, and $D = diag(\sigma_1, \sigma_2,\dots, \sigma_r, 0, \dots,0)$ the syndrome module over $\mathbb{Z}$ is 
    \begin{equation}
        \mathcal{S}_\mathbb{Z} \cong  \bigoplus_{i=1}^{r}(\mathbb{Z/}\sigma_i \mathbb{Z})\oplus \mathbb{Z}^k 
    \end{equation}
    Upon reduction modulo some $p$, we get \begin{equation}
        \mathcal{S}_p \cong \mathcal{S}_{\mathbb{Z}}\otimes \mathbb{F}_p
    \end{equation}
    The fact that $U,V$ are unimodular and non-unique allows us to change basis to $\mathbb{F}_p$ without destroying $U$ or $V$. By the Universal Coefficient Theorem for homology, the rank of $\mathcal{S}$ modulo $p$ is only preserved iff $p$ does not divide any non-zero invariant factor $\sigma_i$. We must also avoid destroying any connections in the stabilizer matrix and therefore need to avoid all primes that divide any entry in the stabilizer matrix. We define the Potential Torsion group $\mathcal{T}$ to be
     \begin{equation}
        \mathcal{T} = \{p : p \ |\ \sigma_i ~\exists i ~or~ p|q_i \in S~\exists i\}
    \end{equation} If chosen basis does not belong to $\mathcal{T}$, $rank(\mathcal{S}_\mathbb{Z}) = rank(\mathcal{S}_p)$. By the rank-nullity theorem, this implies $\dim(\ker(\mathcal{S}_p)) = \dim(\ker(\mathcal{S}_\mathbb{Z}))$. By the Universal Coefficient Theorem, we can conclude $\ker(\mathcal{S}_\mathbb{Z}) \otimes \mathbb{F}_p \cong \ker(\mathcal{S}_p)$.
    This implies that for a modular reduction by a prime not in the torsion set, there are no new low-weight logical operators that arise purely from torsion effects and the homology structure of the syndrome space is preserved, thus $d_{p} \geq d \ \forall p \notin \mathcal{T}$. The injectivity of the map on the generated sub-modules guarantees that the distinguishability of error syndromes is preserved from $\mathbb{Z}$ to $\mathbb{F}_p$. 
    With this, we can define $p^*_{min}$ to be the smallest prime outside the potential torsion group of a code.\begin{equation}
        p_{min}^* = p > p_i \ \forall p_i \in \mathcal{T}
    \end{equation}
    We may compute this over a basis of our syndrome space since SNF invariant factors are independent of choice of basis, lower bounding $p^*$ for a given code.
\end{proof}
By Chinese Remainder Theorem, we can extend this to composite rings of form $Q = \prod^n_{i=0} p_i^1 ~\forall p_i \notin \mathcal{T}$ following the arguments in \cite{gunderman2024stabilizer}.

"""
\fi


\bibliographystyle{unsrt}
\phantomsection  
\renewcommand*{\bibname}{References}

\makeatletter
\NewCommandCopy\LDIoriginalbibitem\bibitem
\RenewDocumentCommand{\bibitem}{o m}{%
  \renewcommand{\bibnumfmt}[1]{%
    \ifcsdef{LDIcount@#2}%
      {\hyper@linkstart{cite}{LDIcite.#2.1}[##1]\hyper@linkend}%
      {[##1]}}%
  \IfNoValueTF{#1}%
    {\LDIoriginalbibitem{#2}}%
    {\LDIoriginalbibitem[#1]{#2}}%
}
\makeatother

\bibliography{main}

\end{document}